\documentclass[a4paper,onecolumn,11pt,accepted=2022-10-17]{quantumarticle}
\pdfoutput=1
\usepackage[utf8]{inputenc}
\usepackage[english]{babel}
\usepackage[T1]{fontenc}
\usepackage{amsmath}
\usepackage{amsfonts}
\usepackage{amsthm}
\usepackage{hyperref}
\usepackage{xcolor}
\usepackage{graphicx}
\usepackage{cite}
\newtheorem{defi}{Definition}
\newtheorem{prop}{Proposition}
\newtheorem{corol}{Corollary}

\begin{document}
\title{Mutually unbiased frames}

\author{Fabian Caro P\'erez}
\affiliation{Departamento de F\'{i}sica, Facultad de Ciencias B\'{a}sicas, Universidad de Antofagasta, Casilla 170, Antofagasta, Chile}
\email{fabian.caro.perez@ua.cl}
\orcid{0000-0001-9861-2491}
\author{Victor Gonz\'alez Avella}
\affiliation{Departamento de F\'{i}sica, Facultad de Ciencias B\'{a}sicas, Universidad de Antofagasta, Casilla 170, Antofagasta, Chile}
\email{victor.gonzalez.avella@ua.cl}
\orcid{0000-0003-2633-6146}
\author{Dardo Goyeneche}
\affiliation{Departamento de F\'{i}sica, Facultad de Ciencias B\'{a}sicas, Universidad de Antofagasta, Casilla 170, Antofagasta, Chile}
\email{dardo.goyeneche@uantof.cl}
\orcid{0000-0002-9865-4226}
\homepage{https://www.quiga.cl/}

\date{May 11, 2022}
\maketitle

\begin{abstract}
In this work, the concept of mutually unbiased frames is introduced as the most general notion of unbiasedness for sets composed by linearly independent and normalized vectors. It encompasses the already existing notions of unbiasedness for orthonormal bases, regular simplices, equiangular tight frames, positive operator valued measure, and also includes symmetric informationally complete quantum measurements. After introducing the tool, its power is shown by finding the following results about the last mentioned class of constellations: \textit{(i)} real fiducial states do not exist in any even dimension, and \textit{(ii)} unknown $d$-dimensional fiducial states are parameterized, a priori, with roughly $3d/2$ real variables only, without loss of generality. Furthermore, multi-parametric families of pure quantum states having minimum uncertainty with regard to several choices of $d+1$ orthonormal bases are shown, in every dimension $d$. These last families contain all existing fiducial states in every finite dimension, and the bases include maximal sets of $d+1$ mutually unbiased bases, when $d$ is a prime number.
\end{abstract}

\section{Introduction}
Finite tight frames \cite{W18,BF03} are a natural generalization of orthonormal bases \cite{KC07}. These frames play a central role in quantum physics, as they are one-to-one related to  rank-one Positive Operator Valued Measures (POVM), a generalized notion of measurements in quantum mechanics. Tight frames find applications in signal theory \cite{BYP04, PY97}, coding and communication \cite{SH03}, design of polarimeters that maximize the signal-to-noise ratio, and reduction of error propagation \cite{T02,CK13}. Optimal ways to implement quantum mechanical protocols typically require to consider sets of quantum measurements induced by vectors with special geometrical properties, e.g. mutually unbiased bases (MUB) \cite{I81,WF89} and symmetric informationally complete (SIC)-POVM \cite{R04}. These constellations are relevant to quantum state estimation \cite{WF89,AS10}, entanglement detection \cite{SHBAH12}, quantum key distribution \cite{MDGGMPKPLF13}, and quantum-nonlocality \cite{BCPSW14,TFRBK21}, among others. On the other hand, sets of frames are relevant in quantum state tomography \cite{PBJD10}, quantum cloning \cite{JRD10} and quantum state discrimination \cite{RSHK11}. 

Along the last years, the notion of mutually unbiasedness has be extended to mutually unbiased simplices \cite{FS20,S19}, mutually unbiased equiangular tight frames \cite{FM20}, and mutually unbiased POVM \cite{BBBCHT13}. Motivated by trying to understand the tree from the exploration of the forest, the most general notion of mutually unbiasedness for constellations composed by linearly independent vectors is introduced, namely \emph{mutually unbiased frames} (MUF). This notion allows to study some geometrical problems from a general perspective, revealing interesting features that are not simple to see otherwise. For instance, a series of new results about SIC-POVM, holding in every finite dimension where they exist, can be found in a simple way when considering MUF. This tool might contribute to unlock hidden symmetries in further classes of constellations, as well as to allow the design of potential practical applications, see Section \ref{sec:sic}). \medskip

This work is organized as follows: in Section \ref{sec:preliminaries}, the basic notions required to understand the work are introduced. In Section \ref{sec:MUF}, our main tool is presented and some basic results are derived. In Section \ref{sec:matrixMUF}, a matrix approach to MUF is shown, specially focusing on circulant matrices. Through this study, both the mutually unbiased bases problem in prime dimension $d$ and the SIC-POVM problem in every dimension $d$ can be seen as $d$ MUF, encoded in $d$ circulant matrices. As a further result, the most general pair of MUF for a qubit system is displayed, emphasizing the fact that both MUB and SIC-POVM designs arise as particular solutions. In Section \ref{sec:sic}, the MUF formalism is applied to the SIC-POVM problem, thus deriving the most important results of this work. Here, the Zauner's conjecture about the existence of Weyl-Heisenberg covariant SIC-POVM is shown to be equivalent to the simplest way to study $d$ MUF in dimension $d$, assuming that circulant matrices define the simplest class of commuting matrices to deal with. Also, it is proven that real fiducial states, with respect to the standard Weyl-Heisenberg (WH) group, do not exist in any even dimension. Moreover, WH covariant fiducial states belong to an explicitly given $\lfloor(d-1)/2\rfloor+d-1$ dimensional set, in every dimension $d$ where they exist. Finally, fiducial states are shown to define minimum uncertainty states regarding a large number of subsets of $d+1$ bases, induced by Clifford unitaries.

\section{Preliminaries}\label{sec:preliminaries}
In this section, the required notions to understand the work are introduced. Firstly, pure quantum states, denoted in kets as $|\psi\rangle$, are normalized vectors in $\mathbb{C}^d$ with regard to the norm induced by the standard scalar product. That is, $\|\psi\|^2=\langle\psi|\psi\rangle=\sum_{j=1}^d|\langle j|\psi\rangle|^2=1$, where $|j\rangle$ denotes the $j$th element of the canonical basis, also called computational basis in quantum information theory. Here, $|\psi\rangle$ is an element in $\mathcal{H}_d$, $\langle\phi|$ an element in the dual space $\mathcal{H}_d^*$, $\langle\phi|\psi\rangle$ and $|\phi\rangle\langle\psi|$ the inner and outer products, respectively. Besides, the symbol $|\phi\rangle\langle\phi|$ denotes the rank-one projector associated to the direction $|\phi\rangle$. The first relevant ingredient is the concept of frames  \cite{DS52}, defined as follows:
\begin{defi}
A set of $n\geq d$ pure states $\{|\phi_i\rangle\}_{i=1,\dots,n}$ in $\mathcal{H}_d$ defines a frame if constants $0<A\leq B<\infty$ exist such that
\begin{equation}\label{frame}
A\leq \sum_{i=0}^{n-1}|\langle x |\phi_i\rangle|^2\leq B,\qquad\mbox{for all }|x\rangle\in\mathcal{H}_d.
\end{equation}
\end{defi}
Here, $A$ and $B$ are called the lower and upper bounds of the frame, given by the minimal and maximal eigenvalues of the \emph{frame operator} $S=\sum_{i=0}^{n-1}|\phi_i\rangle \langle\phi_i|$, respectively \cite{CK13}. A frame is called \emph{tight} if $A$ equals $B$, having in such case that $A=B=n/d$. When $A=B=1$, a Parseval tight Frame arises, equivalent to orthonormal bases in finite dimensions. There is a simple way to detect whether a given set of vectors forms a frame: linearly independence of $d$ out of $n\geq d$ vectors is a necessary and sufficient condition to have a frame in $\mathcal{H}_d$. A gentle introduction to finite frames theory can be found in Ref. \cite{C13}.

A special kind of tight frames occurs when its $n$ vectors are equiangular, known as \emph{equiangular tight frames} (ETF) \cite{STDH07}, meaning that there is a constant $c>0$ such that $|\langle\phi_i|\phi_j\rangle|^2=c$, for every $i\not=j$, with $i,j\in\{0,\dots,n-1\}$. Here, the number $\sqrt{c}$ is typically called \emph{coherence}, but for convenient reasons, throughout this work the constant $c$ is called the \emph{overlap} of a MUF. See \cite{FM15} for a summary of the current state of the art of the ETF problem.

Among all constellations in the Hilbert space $\mathcal{H}_d$, there is a specially distinguished one: mutually unbiased bases \cite{BSTW05}.
\begin{defi}
Two orthonormal bases  $\{|\phi_j\rangle\}_{j=0,\dots,d-1}$ and $\{|\psi_k\rangle\}_{k=0,\dots,d-1}$ defined in $\mathcal{H}_d$ are unbiased if $|\langle\phi_j|\psi_k\rangle|^2=\frac{1}{d}$, for every $j,k=0,\dots,d-1$. Also, a set of $m$ orthonormal bases are MUB if they are pairwise MUB.
\end{defi}
For instance, the eigenvectors bases of the three Pauli matrices define  $m=3$ MUB in $\mathcal{H}_2$. More generally, let \begin{equation}\label{XZ}
X=\sum_{k=0}^{d-1}|k+1\rangle\langle k| \mbox{ and } Z=\sum_{k=0}^{d-1}\omega^k|k\rangle\langle k|,
\end{equation}
the \emph{shift} and \emph{phase} operators, respectively, where $\omega=e^{2\pi i /d}$. Therefore, the eigenvectors bases of the operators $$X,Z,XZ,XZ^2,\dots,XZ^{d-1},$$ form $m=d+1$ MUB in $\mathcal{H}_d$, whenever $d$ is a prime number. Here, and in the rest of the work, it is assumed addition modulo $d$ in kets, i.e. $|(d-1)+1\rangle=|0\rangle$.

MUB have been extensively studied for the last 40 years. A maximal set of $d+1$ MUB exists in every prime \cite{I81} and prime power \cite{WF89} dimension $d$. On the other hand, the maximal number of MUB existing in any other composite dimension $d$ is still an open problem, even in the lowest dimensional case $d=6$, where at most triplets of MUB are known \cite{DEBZ10,RLE11,BBELTZ07,G13}. 

Another remarkable constellations are the SIC-POVM \cite{Z99,RBSC04,SG10}, defined as follows:
\begin{defi}
A set of $n=d^2$ pure states $\{|\phi_j\rangle\}\subset\mathcal{H}_d$ forms a SIC-POVM if $|\langle\phi_j|\phi_k\rangle|^2=\frac{1}{d+1}$, for every $j\not=k$.
\end{defi}
SIC-POVM are known to exist in dimensions $d=2-53$ and in several higher dimensions \cite{G21}, being $d=19603$ achieved recently \cite{ABHGM21}. Also, numerical solutions are known in dimensions $d=2-193$ and some other dimensions up to 2208 \cite{S17,GS17}. In very low dimensions, fiducial states are simple to find. However, as the dimension increases, even the problem to find numerical solutions becomes hard.  On the one hand, the difficulty to solve the MUB and SIC-POVM problems have motivated researchers to find solution to its extensions to the space of full-rank density matrices, see \cite{KG14} and \cite{GK14}, respectively. On the other hand, the rank-one SIC-POVM problem has triggered the search of additional symmetries in fiducial pure states. In short, every known Weyl-Heisenberg fiducial state is eigenvector of an order 3 Clifford operator \cite{A05}. A detailed description of symmetries in fiducial states can be found here \cite{Z99,SG10,S17,A05,BW19}. Among the entire set of SIC-POVM, there is a distinguished one: the \emph{Hoggar lines} \cite{H98}. These solutions for 3-qubit systems are characterized by  fiducial states that are covariant with respect to the tensor product of single qubit WH groups. The 240 fiducials existing in this case can be divided in two classes, distinguished by the fact that they have different amounts of three-partite entanglement \cite{CGZ18}.

The symmetries mentioned above   allowed the construction of every known fiducial state. However, there is \emph{no single proof} for the veracity of such symmetries, beyond its remarkable success. In Section \ref{sec:sic}, it is provided an important cornerstone towards this direction by \textit{revealing an elusive symmetry} that holds for any existing fiducial state. Precisely, the number of free parameters of any fiducial state can be reduced from $2(d-1)$  to $\lfloor (d-1)/2\rfloor+d-1$, in any dimension $d\geq2$. 

To conclude this section, it is remarked the connection existing between the SIC-POVM problem and algebraic number theory, including a close relation to the 12$th$ Hilbert problem \cite{AFMY17}.

\section{Mutually unbiased frames}\label{sec:MUF}
The notion of mutually unbiasedness for constellations of vectors has been extensively used in quantum mechanics. One of the reasons is because \textit{two Von Neumann observables are complementary if and only if its eigenvectors bases are mutually unbiased}, e.g. position and momentum or orthogonal directions of spin $1/2$ observables. Mutually unbiasedness has been extended to regular simplices \cite{FS20,S19}, equiangular tight frames \cite{FM20} and POVM in general \cite{BBBCHT13}. In this section, a notion of unbiasedness that includes all the above notions is introduced. The aim to present this generalization is to study the existence and construction of some inequivalent geometrical structures under the same framework, e.g. MUB and SIC-POVM. Despite our generalization goes beyond the set of quantum measurements, the new point of view reveals interesting properties when restricting the attention to POVM, as shown in Section \ref{sec:sic}. The central concept of our work is defined as follows:
\begin{defi}\label{defi:MUF}
Let $\{|\phi_j\rangle\}_{j=0,\dots,n_1-1}$ and $\{|\psi_k\rangle\}_{k=0,\dots,n_2-1}$ be two frames in $\mathcal{H}_d$. They are called unbiased if there is a constant $c>0$ such that \begin{equation}\label{MUF}
|\langle\phi_j|\psi_k\rangle|^2=c,
\end{equation}
holds for every $j=0,\dots,n_1-1$ and $k=0,\dots,n_2-1$. A set of $m$ frames is called mutually unbiased (MUF) if they are pairwise unbiased. The constant value $c$ is called the overlap of the MUF.
\end{defi}
\begin{figure}[t]
\centering
\includegraphics[width=0.5\textwidth]{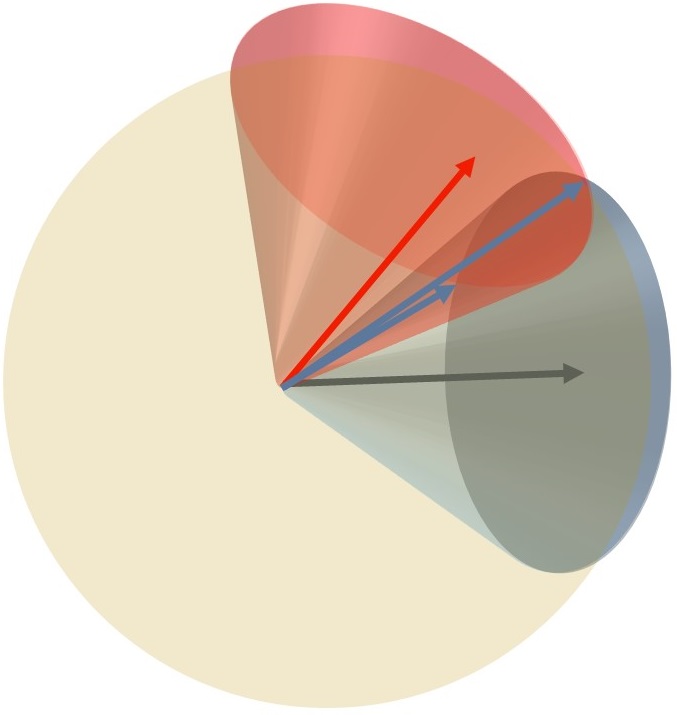}
\caption{(color online) Construction of two mutually unbiased frames composed by two vectors each for a qubit system in the Bloch sphere. Given a frame in $\mathbb{C}^2$, associated to two normalized vectors on the surface of the Bloch sphere, red and gray vectors in the figure, one can always define another frame that is mutually unbiased to the above one as follows: construct two cones with axes given by the red and grey vectors, respectively. The apertures of the cones have to be identical and large enough to have intersection of the cones. The unbiased frame is given by vectors lying at the intersection of the cones, i.e. the blue vectors. When the red and gray vectors are anti-parallel then there is a unique way to define intersecting cones, i.e.  the aperture of each cone is equal to $\pi$, and both cones become the same disk. In such case, any point in the border of the disk is unbiased to the given basis. In this way, 3 MUB arise as 3 orthogonal lines.}
\label{Fig1}
\end{figure}
Fig.\ref{Fig1} shows an intuitive construction of two MUF for a qubit system. The notion of MUF reduces to mutually unbiased POVM \cite{BBBCHT13} when both frames are tight, for any $n_1$ and $n_2$. Additionally, if $n_1=n_2=d$, then they are MUB. Figure \ref{Fig2} illustrates the relation existing between different 
notions of unbiasedness.

 At first glance, a set of $m$ MUF might look as a collection of equiangular lines. However, it is important to remark that the MUF structure has additional constraints, as the following example shows. Suppose that $4$ MUF exist in dimension $3$, each of them having $3$ elements, and arranged along 4 rows as follows: 
\begin{equation*}
\begin{array}{cccc}
a_1&a_2&a_3&\quad\mbox{$\leftarrow$ Frame 1}\\
b_1&b_2&b_3&\quad\mbox{$\leftarrow$ Frame 2}\\
c_1&c_2&c_3&\quad\mbox{$\leftarrow$ Frame 3}\\
d_1&d_2&d_3&\quad\mbox{$\leftarrow$ Frame 4}
\end{array}
\end{equation*}
Clearly, $\{a_i,b_i,c_i,d_i\}$, is a set of 4 equiangular lines for any $i=1,2,3$, which might tempt one to claim that the above set of three equiangular lines is equivalent to a set of four MUF. However, the MUF constraint establishes that two elements of the above arrangement belonging to different rows and different columns (e.g. $a_1$ and $b_2$), also have the same overlap $c$. Indeed, the MUF constraint is equivalent to say that  $\{a_i,b_j,c_k,d_l\}$ forms a set of equiangular lines, for any $i,j,k,l=1,\dots,3$. This is an interesting combinatorial extension of the equiangular lines problem.

Given that all sub-normalized rank-one projectors associated to a tight frame sum up to the identity, it is simple to show that $c=1/d$ holds in (\ref{MUF}) when both frames are tight \cite{BBBCHT13}. However, if the frames are not tight, the constant $c$ might take another value, as shown below.
\begin{prop}\label{prop:muf}
Let $\{|\phi^{(\ell)}_j\rangle\}$ be a set of $m$ mutually unbiased frames in $\mathcal{H}_d$ composed by $n$ elements each, with overlap $c$. Also, assume that each frame has lower and upper frame bounds $A_{\ell}$ and $B_{\ell}$, respectively, where $\ell=0,\dots,m-1$. Therefore, 
 $\max_{\ell} A_{\ell} \leq cn\leq \min_{\ell}B_{\ell}$.
\end{prop}
\begin{proof}
From considering $|x\rangle=|\phi_i^{\ell}\rangle$, inequalities (\ref{frame}) imply
\begin{equation}
A_{1}\leq \sum_{j=0}^{n-1} |\langle \phi_{i}^{(\ell)}|\phi_{j}^{(\ell')}\rangle|^2
\leq B_{1},\, 
A_{2}\leq \sum_{i=0}^{n-1} |\langle \phi_{i}^{(\ell)}|\phi_{j}^{(\ell')}\rangle|^2
\leq B_{2},\dots,\, A_{m} \leq \sum_{j=0}^{d-1} |\langle \phi_{i}^{(\ell)}|\phi_{j}^{(\ell')}\rangle|^2
\leq  B_{m},
\end{equation}
for each $\ell=0,\dots,n-1$, which reduce to 
$A_{1} \leq cn \leq B_{1}, A_{2} \leq cn \leq B_{2} ,\dotsc, A_{m} \leq cn \leq B_{m}$. Therefore, the most restrictive constraint for $c$ is given by $\max_{\ell} A_{\ell}\leq cn\leq \min_{\ell}B_{\ell}$.
\end{proof}

\begin{figure}[htp]
\centering
\includegraphics[width=0.7\textwidth]{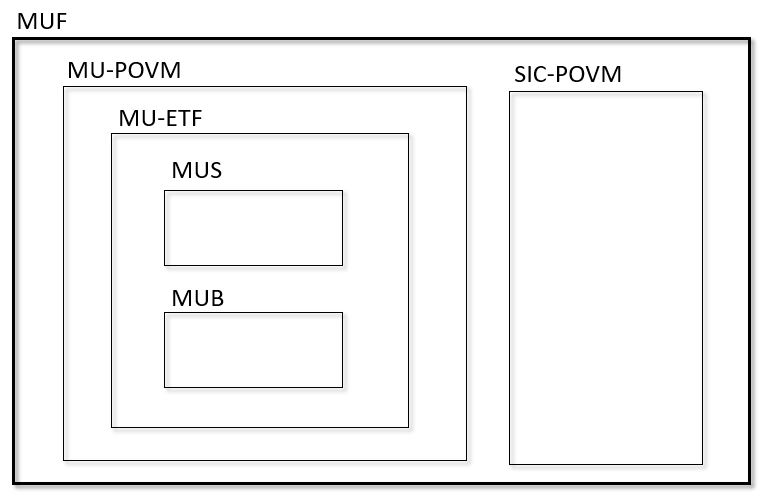}
\caption{Mutually unbiased frames (MUF), the main concept introduced in this work, represents the most general notion of unbiasedness for sets of linearly independent vectors. It generalizes the existing notions of unbiasedness for orthonormal bases (MUB) \cite{I81,WF89}, simplices (MUS) \cite{FS20,S19}, equiangular tight frames (MU-ETF) \cite{FM20} and POVM (MU-POVM) \cite{BBBCHT13}. Additionally, SIC-POVM \cite{Z99,RBSC04} can be seen as sets of $d$ MUF in dimension $d$, as shown in Section \ref{sec:sic}.}
\label{Fig2}
\end{figure}
For the case of orthonormal bases, Proposition \ref{prop:muf} reduces to mutually unbiased bases, where $c=1/d$. Indeed, the same overlap $c$ occurs when at least one of the frames is tight, thus forming a POVM.
\begin{corol}\label{corol:muf}
Let $\{|\phi^{(\ell)}_j\rangle\}$ be a set of $m$ mutually unbiased frames in $\mathcal{H}_d$ with overlap $c$, composed by $n$ elements each, such that at least one of them is tight. Therefore, 
$c=1/d$.
\end{corol}
\begin{proof}
The proof arises from Proposition \ref{prop:muf}, where $\frac{n}{d}=A_{\ell'}=B_{\ell'}=cn$ holds. Thus, $c=1/d$.
\end{proof}

From Corollary \ref{corol:muf}, note that Parseval MUF are mutually unbiased bases in every finite dimension $d$. In Section \ref{sec:matrixMUF}, examples of MUF in $\mathcal{H}_d$ are shown, where the overlaps $c$ are not equal to $1/d$, in general.

\section{Matrix approach to mutually unbiased frames}\label{sec:matrixMUF}
In this section, a matrix approach to the problem of mutually unbiased frames is shown. In particular, the case of $m$ MUF composed by $d$ vectors each in $\mathcal{H}_d$ is considered. Through this approach, entries of the $j$-th element of each frame are arranged as columns belonging to a square complex matrix of order $d$, for every $j=0,\dots,d-1$. Thus, there are $m$ square matrices of order $d$, denoted $M_0,\dots,M_{m-1}$, satisfying the following relations:
\begin{equation}\label{matrices_muf}
M_j^{\dag}M_k=K^{(jk)},
\end{equation}
where $|(K^{(jk)})_{st}|^2=c$, for every $s,t=0,\dots,d-1$, and $j\not=k$. Here, $(K^{(jk)})_{st}$ denotes the entry of matrix $K^{(jk)}$ located in row $s$ and column $t$. Also, if $j=k$ the relation $(K^{(jj)})_{ss}=1$ occurs, since the column vectors of $M_j$ are normalized, for every $j=0,\dots,m-1$, and $s=0,\dots,d-1$. As a general property, it is simple to show that $M_j$ and $K^{(jk)}$ are full rank matrices, as their columns have to be linearly independent in order to form a frame. However, the problem to find sets of $m$ matrices $\{M_j\}$ satisfying (\ref{matrices_muf}) is evidently hard, if no further assumptions are taken. In order to simplify this problem, the following symmetries are assumed: matrices $M_j$ are normal, i.e. $[M^{\dag}_j,M_j]=0$, and they do commute, i.e. $[M_j,M_k]=0$ and $[M_j,M_k^{\dag}]=0$. These assumptions imply that $[M_{\ell},K^{(jk)}]=0$ is satisfied, for every $j,k,\ell=0,\dots,d-1$. Thus, the problem simplifies to finding special relations existing between the spectrum of matrices $M_j$ and $K^{(jk)}$.

At this stage, the choice of the unitary transformation $U$ that diagonalizes all matrices $M_j$ plays a fundamental role. For instance, if matrices $M_j$ are diagonal in the computational basis, then there is no pair of MUF. Thus, an interesting question arises: \emph{which choices of $U$ provide a solution to (\ref{matrices_muf})} for a given value of $m\geq2$? This question is hard to solve, and it is not studied through the current research. In turns, the simplest case of matrices $M_j$ diagonalized by the discrete Fourier transform is considered. That is, the matrices $M_0,\dots,M_{m-1}$ are circulant. This choice is justified by simplicity, as several useful properties are well-known for this class of matrices \cite{D98}. A square matrix $M$ is called circulant if the $j+1$-th row is a cyclic permutation of the $j$-th row, for every $j=0,\dots-d-2$. That is,
\begin{equation}
M=\left(\begin{array}{ccccc}
m_0&m_1&m_2&\dots&m_{d-1}\\
m_{d-1}&m_0&m_1&\dots&m_{d-2}\\
\vdots&&&&\vdots\\
m_1&m_2&m_3&\dots&m_0
\end{array}\right),
\end{equation}
where $m_0,\dots,m_{d-1}\in\mathbb{C}$.
Let $\vec{\mu}(M)$ be the vector formed by all elements of the first row of $M$, i.e.   $\vec{\mu}=(m_0,\dots,m_{d-1})^T$, where $T$ denotes transposition. Also, let $\vec{\lambda}(M)$ be a column vector containing all  eigenvalues of $M$, sorted according to the decomposition
\begin{equation}\label{M}
M=F\mathrm{diag}[\vec{\lambda}(M)]F^{\dag},
\end{equation}
where $\mathrm{diag}[\vec{\lambda}(M)]$ is a diagonal matrix containing the vector $\vec{\lambda}(M)$ along its main diagonal, and $F$ is the discrete Fourier transform, having entries $F_{jk}=\frac{1}{\sqrt{d}}\omega^{jk}$, where  $\omega=e^{2\pi i/d}$. Along the work, $\vec{\lambda}(M)$ is called \textit{the vector defined by the eigenvalues of $M$.} From (\ref{M}), it is simple to show that $\vec{\lambda}(M)=\sqrt{d}F^{\dag}\vec{\mu}(M)$ is satisfied. On the other hand, the unitary operator $F^2$ transforms the first row of $M$, i.e. $\vec{\mu}(M)$, to the first column of $M$, that has unit norm. Therefore, normalized columns in $M$ translates to the condition $\|\vec{\mu}(M)\|=1$, or equivalently  $\|\vec{\lambda}(M)\|=\sqrt{d}$. 

Similarly, any matrix $K$ from (\ref{matrices_muf}) is circulant, having the form
\begin{equation}
K=\sqrt{c}\left(\begin{array}{ccccc}
e^{i\alpha_0}&e^{i\alpha_1}&e^{i\alpha_2}&\dots&e^{i\alpha_{d-1}}\\
e^{i\alpha_{d-1}}&e^{i\alpha_0}&e^{i\alpha_1}&\dots&e^{i\alpha_{d-2}}\\
\vdots&&&&\vdots\\
e^{i\alpha_1}&e^{i\alpha_2}&e^{i\alpha_3}&\dots&e^{i\alpha_{0}}
\end{array}\right),
\end{equation}
where $\alpha_0,\dots,\alpha_{d-1}\in[0,2\pi)$ and $0<c<1$. The first row of $K$ is arranged as the vector $\vec{\mu}(K)=\sqrt{c}(e^{i\alpha_0},\dots,e^{i\alpha_{d-1}})^T$, and the vector $\vec{\lambda}(K)$ containing all eigenvalues of $K$ is defined according to
\begin{equation}\label{K}
K=F\mathrm{diag}[\vec{\lambda}(K)]F^{\dag}.
\end{equation}
As a basic property of circulant matrices, the relation $\vec{\lambda}(K)=\sqrt{d}F^{\dag}\vec{\mu}(K)$ occurs, where $\|\vec{\mu}(K)\|=\sqrt{cd}$ and $\|\vec{\lambda}(K)\|=d\sqrt{c}$. Now, the first result of this section is established.

\begin{prop}\label{equi_M}
Let $M_0,\dots,M_{m-1}$ be $m$ circulant matrices that form $m$ MUF composed by $d$ vectors each, in dimension $d$. Then, the normalized vectors defined by the eigenvalues of these matrices form $m$ equiangular lines in $\mathbb{C}^d$, with overlap $c$.
\end{prop}
\begin{proof}
The result arises from direct calculation:
\begin{eqnarray}
|\vec{\lambda}(M_j)\cdot\vec{\lambda}(M_k)|^2&=&|\mathrm{Tr}\bigl(\mathrm{diag}[\vec{\lambda}(M_j)]^{\dag}\mathrm{diag}[\vec{\lambda}(M_k)]\bigr)|^2\nonumber\\
&=&|\mathrm{Tr}(F^{\dag}\vec{\lambda}(M_j)]^{\dag}FF^{\dag}\vec{\lambda}(M_k)]F)|^2\nonumber\\
&=&|\mathrm{Tr}(M_{j}^{\dag}M_{k})|^2\nonumber\\
&=&|\mathrm{Tr}(K)|^2\nonumber\\
&=&cd,
\end{eqnarray}
for every $j\not=k\in\{0,\dots,m-1\}$. Moreover, when eigenvectors $\vec{\lambda}(M_0),\dots,\vec{\lambda}(M_{m-1})$ are normalized, they define $m$ equiangular lines with overlap $c$, the same overlap of the MUF.
\end{proof}
An expert reader might recognize here the evidence of a well-known property holding for fiducial states of SIC-POVM: the Fourier transform of a fiducial state is a fiducial state. From combining Prop. \ref{equi_M} with results from equiangular lines theory \cite{GKMS}, the following constraints for the existence of MUF is obtained:
\begin{prop}
Suppose that $M_1,\dots,M_m$ are circulant matrices of order $d$, forming $m$ MUF. Then, the inequality  $m\leq d^2$ holds.
\end{prop}
\begin{proof}
The proof arises from the fact that there are no more than $d^2$ equiangular lines in dimension $d$.
\end{proof}
Moreover, stronger bounds can be found when the eigenvalues of $M_j$ are real.
\begin{prop}
Suppose a set MUF composed of $m$ hermitian circulant matrices, with overlap $c$, exists. Therefore, the following properties hold:
\begin{enumerate}
\item $m\leq d(d+1)/2$ (see \cite{LS73}).
\item  If $m\geq2d$, then $1/\sqrt{c}$ is an integer number (see \cite{LS73}).
\item If $c\leq1/(d+2)$, then $m\leq d(1-c)/(1-dc)$ (see Lemma 6.1 here \cite{LS66}). 
\item If $c\leq1/d^2$, then $m\leq d+1$.
\end{enumerate}
\end{prop}
\begin{proof}
Proofs are based on the fact that eigenvectors of all matrices forming the MUF are real (due to hermiticity) and equiangular (due to Prop \ref{equi_M}). Thus, the constraints established in \emph{1.}, \emph{2.} and \emph{3.} occur for real ETF, which can be found in the described references. Item \emph{4.} arises from assuming that $c\leq1/d^2$, and then using \emph{3.} 
\end{proof}
Further properties of eigenvalues of matrices $M_j$ can be derived. By using (\ref{matrices_muf}), it occurs that
\begin{equation}
\vec{\lambda}^*(M_j)\circ\vec{\lambda}(M_k)=\vec{\lambda}(K),
\end{equation}
where $\vec{\lambda}(M_j)$ is the vector associated to $M_j$, the circle ($\circ$) denotes the entrywise (Hadamard) product and the asterisk (*) complex conjugation in the computational basis. Note that the inequality norm  $\|\vec{v}_1\circ\vec{v}_2\|\leq\|\vec{v}_1\|\|\vec{v}_2\|$, produces $c\leq1$. This upper bound is tight, as two matrices $M_1$ and $M_2$ with normalized and linearly independent columns might be arbitrarily close to a rank one projector $P$, thus producing $c\rightarrow1$. Also, it might occur that  $M_1\rightarrow P_1$ and $M_2\rightarrow P_2$, with orthogonal projectors $P_1$ and $P_2$ satisfying $P_1P_2=0$, thus having $c\rightarrow0$.

To illustrate the above results, the most general family of 2 MUF in dimension 2, arising from circulant matrices, is derived as follows. Suppose that
\begin{eqnarray}\label{M1M2_MUF}
M_{1}=
\begin{pmatrix}
\cos(\theta) & \sin(\theta)e^{i\alpha} \\
\sin(\theta)e^{i\alpha} & \cos(\theta)
\end{pmatrix}\quad\mbox{and}\quad M_{2}=
\begin{pmatrix}
\cos(\eta) & \sin(\eta)e^{i\beta} \\
\sin(\eta)e^{i\beta}& \cos(\eta)
\end{pmatrix},
\end{eqnarray}
satisfy the relations $M^{\dag}_1M_2=K$, with $|K_{ij}|^2=c$, $i,j=0,\dots,d-1$. From a straight calculation, a sufficient condition to have a pair of MUF from (\ref{M1M2_MUF}) is obtained:
\begin{eqnarray}\label{sol_muf_d2}
\sin(\alpha)\sin(\beta)\tan(2\theta)\tan(2\eta)=-1.
\end{eqnarray}
Note that (\ref{sol_muf_d2}) has a solution if and only if the following properties hold:
\begin{equation*}
\alpha,\beta,\theta,\eta\not\in\{0,\pi/2,\pi,3\pi/2\},
\end{equation*}
\begin{equation*} -1\leq\sin^{-1}(\alpha)\tan^{-1}(2\theta)\tan^{-1}(2\eta)\leq1,
\end{equation*}
and
\begin{equation*}
-1\leq\sin^{-1}(\beta)\tan^{-1}(2\theta)\tan^{-1}(2\eta)\leq1,
\end{equation*}
where the negative power means inverse multiplicative. This generic solution determines a 3-parametric nonlinear family of MUF with overlap
\begin{equation}
c=\cos   ^2(\eta-\theta)+2\sin(\eta)\cos(\eta)\sin (\theta)\cos(\theta)(\cos(\alpha-\beta)-1),
\end{equation}
which includes SIC-POVM when, for instance, $\alpha=\beta$ and $\theta-\eta=\arccos(1/\sqrt{3})$. Note that there is no real pair of MUF within this continuous family. Nonetheless, real pairs of MUF arise from (\ref{M1M2_MUF}) when  $\theta\in\{0,\pi\}$ and $\eta\in\{\pi/2,3\pi/2\}$, also when $\theta\in\{\pi/2,3\pi/2\}$ and $\eta\in\{0,\pi\}$, and finally when $\alpha,\beta\in\{\pi/2,3\pi/2\}$. All these isolated solutions form pairs of real MUB.

Interestingly, the above study reveals that MUF based on circulant matrices admit solution to both the SIC-POVM and the MUB problem for a qubit system. The same occurs for every prime dimension $d$: the eigenvectors bases of operators $Z,XZ,\dots,XZ^{d-1}$, forming a maximal set of $d+1$ MUB, can be arranged as $d$ circulant unitary matrices of order $d$ \cite{BBRV01}. Moreover, Zauner's conjecture \cite{Z99} allows to see a SIC-POVM as a set of MUF composed by $d$ circulant matrices, as it will be shown in Section \ref{sec:sic}. Despite the existence of $m$ MUF in dimension $d$ is a hard problem, a solution for the case of $m=2$ is shown below, in any dimension $d$.
\begin{prop}\label{2MUF}
Let $M$ be a circulant matrix of order $d$ with columns forming a normalized frame. Then, there is another matrix $\tilde{M}$ inducing a frame such that $M$ and $\tilde{M}$ form a pair of MUF, in every dimension $d$. 
\end{prop}
\begin{proof}
Matrices $M$ and $\tilde{M}$ are MUF if and only if $M^{\dag}\tilde{M}=K$, with $|K_{jk}|^2=c$, for every $j,k=0,\dots,d-1$  and some $0<c<1$. This is equivalent to say that their normalized vectors defined by its eigenvalues satisfy
\begin{equation}\label{lambdas_app}
\vec{\lambda}^*_M\circ\vec{\lambda}_{\tilde{M}}=\vec{\lambda}_{K}.
\end{equation}
By assuming the polar decompositions
\begin{equation}
\lambda_{M}=\sum_{j=0}^{d-1}r_j\omega^{\alpha_j}|j\rangle\quad\mbox{and}\quad\lambda_{\tilde{M}}=\sum_{j=0}^{d-1}s_j\omega^{\beta_j}|j\rangle,
\end{equation}
where $r_j,s_j,\alpha_j,\beta_j\in\mathbb{R}$, $j=0,\dots,d-1$, and the fact that
\begin{equation}
\vec{\lambda}_{K}=\sqrt{c}\,F^{\dag}\sum_{j=0}^{d-1}\omega^{\gamma_j}|j\rangle,
\end{equation}
for some $\gamma_0,\dots,\gamma_{d-1}\in[0,2\pi)$, one can re-write (\ref{lambdas_app}) as
\begin{equation}\label{polar_rs_app}
\sum_{j=0}^{d-1}r_js_j\omega^{\beta_j-\alpha_j}|j\rangle=\sqrt{cd}\,F^{\dag}\sum_{j=0}^{d-1}\omega^{\gamma_j}|j\rangle.
\end{equation}
Note that $r_j\not=0$ and $s_j\not=0$, for any $j=0,\dots,d-1$. This is so because eigenvalues of matrices $M$ and $\tilde{M}$ cannot vanish, because their columns form a frame. Now, if there is an -unormalized- vector $\sum_{j=0}^{d-1}\omega^{\gamma_j}|j\rangle$ unbiased to the pair $\{\mathbb{I},F\}$, then there exists $\delta_0,\dots,\delta_{d-1}\in[0,2\pi)$ such that $F^{\dag}\sum_{j=0}^{d-1}\omega^{\gamma_j}|j\rangle=\sum_{j=0}^{d-1}\omega^{\delta_j}|j\rangle$. So, (\ref{polar_rs_app}) reduces to
\begin{equation}
\sum_{j=0}^{d-1}r_js_j\omega^{\beta_j-\alpha_j}|j\rangle=\sqrt{cd}\,\sum_{j=0}^{d-1}\omega^{\delta_j}|j\rangle,
\end{equation}
implying that $r_js_j=\sqrt{cd}$ and $\beta_j-\alpha_j=\delta_j$, for every $j=0,\dots,d-1$. Also, recall that normalization of vectors of the frames imply the additional constraints $\sum_{j=0}^{d-1}r^2_j=\sum_{j=0}^{d-1}s^2_j=d$. Thus, the suitable matrix $\tilde{M}$ is characterized by $s_j=\sqrt{cd}/r_j$ and $\beta_j=\alpha_j+\gamma_j$, with the constraint $\sum_{j=0}^{d-1}s^2_j=cd\sum_{j=0}^{d-1}\frac{1}{r^2_j}=d$, implying that $c=1/\sum_{j=0}^{d-1}\frac{1}{r^2_j}<1$. To conclude the proof, note that there are at least $d$ vectors MU to the pair $\{\mathbb{I},F\}$ in every dimension $d$, corresponding to the eigenvectors of $XZ$.
\end{proof}
Prop. \ref{2MUF} can be illustrated with a 3-parametric family of MUF for a qubit system. Consider the vector $v=(1,i)^T$, which is MU to the pair $\{\mathbb{I},F\}$ in dimension $d=2$. Thus, taking $r_0=\sqrt{2}\cos(\nu)$ and $r_1=\sqrt{2}\sin(\nu)$, it holds that
\begin{eqnarray*}
M&=&Fdiag[\vec{\lambda}(M)]F^{\dag}\\
&=&\frac{1}{2}\left(\begin{array}{cc}
1&1\\
1&-1
\end{array}\right)
\left(\begin{array}{cc}
\sqrt{2}\cos(\nu)\omega^{\alpha_0}&0\\
0&\sqrt{2}\sin(\nu)\omega^{\alpha_1}
\end{array}\right)
\left(\begin{array}{cc}
1&1\\
1&-1
\end{array}\right)\\
&=&\frac{1}{\sqrt{2}}\left(\begin{array}{cc}
\cos(\nu)\omega^{\alpha_0}+\sin(\nu)\omega^{\alpha_1}&\cos(\nu)\omega^{\alpha_0}-\sin(\nu)\omega^{\alpha_1}\\
\cos(\nu)\omega^{\alpha_0}-\sin(\nu)\omega^{\alpha_1}&\cos(\nu)\omega^{\alpha_0}+\sin(\nu)\omega^{\alpha_1}
\end{array}\right),
\end{eqnarray*}
and
\begin{eqnarray*}
\tilde{M}&=&Fdiag[\vec{\lambda}(\tilde{M})]F^{\dag}\\
&=&\frac{1}{2}\left(\begin{array}{cc}
1&1\\
1&-1
\end{array}\right)
\left(\begin{array}{cc}
\frac{\sqrt{c}}{\cos(\nu)}\omega^{\alpha_0+1/4}&0\\
0&\frac{\sqrt{c}}{\sin(\nu)}\omega^{\alpha_1-1/4}
\end{array}\right)
\left(\begin{array}{cc}
1&1\\
1&-1
\end{array}\right)\\
&=&\frac{1}{2} \sqrt{c}\left(
\begin{array}{cc}
 \sec(\nu ) w^{\alpha_0+\frac{1}{4}}+\csc(\nu )
   w^{\alpha_1-\frac{1}{4}} &  w^{\alpha_0+\frac{1}{4}}\sec(\nu)-w^{\alpha_1-\frac{1}{4}} \csc (\nu ) \\
  w^{\alpha_0+\frac{1}{4}} \sec (\nu )-w^{\alpha_1-\frac{1}{4}} \csc (\nu )& \sec (\nu )
   w^{\alpha_0+\frac{1}{4}}+\csc (\nu ) w^{\alpha_1-\frac{1}{4}}
\end{array}
\right),
\end{eqnarray*}
with overlap $c=2\cos^2(\nu)\sin^2(\nu)$. In particular, for $\nu=\pi/4$ and $\nu=-\arctan \left(\sqrt{\frac{3-\sqrt{3}}{3+\sqrt{3}}}\right)$, two MUB and a SIC-POVM occur, respectively.

\section{Application to the SIC-POVM problem}\label{sec:sic}
Along this section, SIC-POVM are organized as $d$ sets of mutually unbiased frames, composed by $d$ vectors each. This slightly different way of seeing these constellations carries out remarkably interesting consequences, as shown below. Based on the matrix approach to MUF, introduced in Section \ref{sec:matrixMUF}, one can assume that the $d^2$ vectors of a SIC-POVM can be arranged as column vectors in $m=d$ circulant matrices of order $d$, $M_0,\dots,M_{d-1}$, each of them defining a frame. This assumption reduces the problem to finding a single suitable matrix of order $d$, called $\mathcal{M}$, that contains the normalized vectors $\vec{\lambda}(M_0)/\sqrt{d},\dots,\vec{\lambda}(M_{d-1})/\sqrt{d}$ along its columns. The problem of finding a solution $\mathcal{M}$ is hard, since it has a quadratic number of free parameters as a function of the dimension $d$. In order to make the problem simpler, a further reasonable assumption is adopted: the matrix $\mathcal{M}$ is circulant. Thus, all the problem reduces to find a suitable first row of $\mathcal{M}$, that has $2(d-1)$ free parameters without loss of generality. After these intuitive, but drastic assumptions, one might wonder whether there is a SIC-POVM satisfying the above symmetries. Below, a positive answer is provided in every dimension $d$ where a SIC-POVM is known.
\begin{prop}\label{MZauner}
The assumption that matrices $M_0,\dots,M_{d-1}$ and $\mathcal{M}$ are circulant is equivalent to the assumption that a SIC-POVM can be defined as an orbit of a fiducial state under the Weyl-Heisenberg group.
\end{prop}
\begin{proof}
The fact that the $d^2$ vectors of a SIC-POVM are arranged as columns in circulant matrices $M_0,\dots,M_{d-1}$ imply that the $d$ vectors within each of these matrices are connected by a shift of its entries. Evenmore, the fact that $\mathcal{M}$ is circulant means that $\vec{\lambda}(M_{j})=X^j\vec{\lambda}(M_0)$, implying that $M_j=Z^jM_0(Z^j)^{\dag}$. Thus, the $d^2$ vectors of the SIC-POVM are given by $|\phi_{jk}\rangle=X^jZ^k|\phi_{00}\rangle$, where $|\phi_{00}\rangle$ is an unknown fiducial state.
\end{proof}
Note that normalized  vectors $\vec{\lambda}(M_j)/\sqrt{d}$, $j=0,\dots,d-1$, are also fiducial states. This is simple to see, as the first row of $M_j$ is connected with the normalized vector $\vec{\lambda}(M_j)/\sqrt{d}$ via the Fourier transform $F$, and the first row of $M_j$ is connected with its first column via $F^2$. Given that $F$ belongs to the Clifford group, it transforms fiducial states into fiducial states. Note that this result is in agreement with Prop. \ref{equi_M}.

A fundamental ingredient to have a frame is the linearly independence of its elements. In the matrix approach, this is equivalent to say that matrices $M_0,\dots,M_{d-1}$ do not have vanishing determinant, i.e. the  vector $\lambda(M_0)$ has non-zero entries. As this vector is also a fiducial state, up to a rescaling factor, it is interesting to wonder whether a fiducial state can have a vanishing entry. Zauner already noted the existence of such fiducial states in dimension 3 \cite{Z99}. Some recent explorations revealed the existence of fiducials with zero entries in dimensions $d=26, 28, 62, 98$ and $228$ \cite{Markus_comment}. For an extended study about linearly dependencies of elements of a SIC-POVM see here \cite{DBBA}. When having a fiducial state $|\phi\rangle$ with zero entries one might consider another fiducial  $U|\phi\rangle$, where $U$ is an element of the Clifford group. Most likely, many of these transformations produce a fiducial state having non-zero entries but, unfortunately, it was not possible to demonstrate it. This property is left as a conjecture. 

There is a further symmetry that one can impose in order to reduce the number of free parameters in SIC-POVM. One might assume that all matrices $M_0,\dots,M_{d-1}$ are hermitian, which holds if and only if $\vec{\lambda}(M_0)/\sqrt{d}$ is a real fiducial state. Along this direction, it has been conjectured that real fiducial states exist in odd dimensions of the form $d=4n^2+3$, $n\in\mathbb{Z}$  \cite{S17}. This conjecture is supported with 18 real fiducials existing in dimensions of such form, the highest dimensional case corresponding to $d=19603$ \cite{ABHGM21}. For instance, in $d=3$, there is a family of real fiducial states \cite{RBSC04}:
\begin{equation}
|\phi(r_0)\rangle=r_0|0\rangle-(r_0 /2 + \sqrt{2-3r_0^2}/2)|1\rangle- (r_0/2-\sqrt{2-3r_0^2}/2)|2\rangle,
\end{equation}
with $1/\sqrt{2}<r_0<\sqrt{2/3}$, inducing the following 1-parametric family of hermitian matrices:
\begin{eqnarray*}
M_j&=&\left(
\begin{array}{ccc}
 0 & a_j & a_j^*\\
 a_j^*&0&a_j\\
 a_j&a_j^*&0
\end{array}
\right),\quad j=0,1,2,
\end{eqnarray*}
where 
\begin{eqnarray*}
a_0&=&\frac{1}{2} \left(\sqrt{3} r_0+i \sqrt{2-3 r_0^2}\right)\nonumber\\ a_1&=&\frac{1}{12} \left(3-i \sqrt{3}\right) \left(\sqrt{6-9 r_0^2}-3 i r_0\right)\nonumber\\
a_2&=&-\frac{1}{12} \left(3+i \sqrt{3}\right) \left(\sqrt{6-9 r_0^2}-3 i r_0\right).
\end{eqnarray*}
Note that all existing real fiducial states are defined in odd dimensions \cite{S17,ABHGM21,K08}. Thus, an interesting question arises: \emph{are there real fiducial states in even dimensions $d$?} Below, a conclusive answer to this question is provided.
\begin{prop}\label{real_fid}
Real fiducial states, covariant under the standard Weyl-Heisenberg group, do not exist in any even dimension.
\end{prop}
\begin{proof}
Let $|\phi\rangle=\sum_{j=0}^{d-1}c_j|\varphi_j\rangle$ be a real pure state, expanded in the Fourier basis $|\varphi_j\rangle=\frac{1}{\sqrt{d}}\sum_{k=0}^{d-1}\omega^{jk}|k\rangle$, where $\omega=e^{2\pi i/d}$. Taking in to account that  $\langle\phi|=\sum_{j,k}c_j\omega^{jk}\langle k|$, it occurs that
\begin{eqnarray*}
|\langle\phi|X^{d/2}Z|\phi\rangle|^2&=&\frac{1}{d}\sum_{j, j',k,k'}c_j\omega^{jk}\langle k|\sum_{\ell=0}^{d-1}|\ell+d/2\rangle\langle\ell|\sum_{m=0}^{d-1}\omega^m|m\rangle\langle m|c_{j'}\omega^{j'k'}|k'\rangle\nonumber\\
&=&\frac{1}{d}\sum_{j, j',k,k',\ell=0}^{d-1}c_jc_{j'}\omega^{jk+j'k'+\ell}\langle k|\ell+d/2\rangle\langle\ell|k'\rangle\nonumber\\
&=&\sum_{j=0}^{d-1}(-1)^jc_jc_{d-j-1}\nonumber\\
&=&(c_0c_{d-1}+c_2c_{d-3}+\dots+c_{d-4}c_{3}+c_{d-2}c_{1})-\\&&(c_1c_{d-2}+c_3c_{d-4}+\dots+c_{d-3}c_{2}+c_{d-1}c_{0})\nonumber\\
&=&0,
\end{eqnarray*}
for every even dimension $d$. Note that for SIC-POVM, the above overlap should be equal to $1/(d+1)$. As $|\phi\rangle$ is an arbitrary real state, there is no real fiducial state with respect to the standard WH group in any even dimension $d$.
\end{proof}
Now, a hidden symmetry is revealed for the amplitudes of fiducial states.
\begin{prop}\label{prop:lambdas}
Any fiducial state $|\phi\rangle$, with respect to a Weyl-Heisenberg covariant SIC-POVM, satisfies
\begin{equation}\label{reduction}
\sum_{j=0}^{d-1}|\langle j|\phi\rangle|^2|j\rangle=F^{\dag}\left(\frac{1}{\sqrt{d}}|0\rangle+\frac{1}{\sqrt{d(d+1)}}\sum_{j=1}^{d-1}e^{i\alpha_j}|j\rangle\right),
\end{equation}
in every finite dimension $d$. Here, $F$ is the discrete Fourier transform of order $d$ and $\alpha_j\in[0,2\pi)$ satisfies $\alpha_{d-j}=-\alpha_j$, for every $j=1,\dots,d-1$. As a consequence, any fiducial state $|\phi\rangle$ depends on $\lfloor (d-1)/2\rfloor+d-1$ free parameters only, i.e., $\lfloor (d-1)/2\rfloor$ amplitudes given by (\ref{reduction}) and $d-1$ complex phases.
\end{prop}
\begin{proof}
Matrix $M_0$ satisfies \begin{equation}\label{M0K_app}
M^{\dag}_0M_0=\tilde{K},
\end{equation} where $\tilde{K}$ is a circulant matrix satisfying
\begin{equation}\label{entriesKtilde_app}
|\tilde{K}_{jk}|^2=\left\{\begin{array}{cc}
1&j=k\\
\frac{1}{d+1}&j\not=k
\end{array}\right..
\end{equation}
Given that $\tilde{K}=F\mathrm{diag}[\vec{\lambda}(\tilde{K})]F^{\dag}$, from (\ref{lambdas_app}) and (\ref{M0K_app}) it occurs that $\vec{\lambda}(M_0)^*\circ\vec{\lambda}(M_0)=\vec{\lambda}(\tilde{K})$. Here, the circle $(\circ)$ denotes the -entrywise- Hadamard product. On the other hand, given that $\tilde{K}$ is circulant, its eigenvalues are related with its first row via the Fourier transform. Thus, from (\ref{entriesKtilde_app}) it holds that
\begin{equation}\label{lambda_k_app}
\vec{\lambda}(\tilde{K})=\sqrt{\frac{d}{d+1}}F^{\dag}\{\sqrt{d+1},e^{i\alpha_1},\dots,e^{i\alpha_{d-1}}\}^T,
\end{equation}
where the fact that the entries of the first row of $\tilde{K}$ are given by $\tilde{K}_{00}=1$ and $\tilde{K}_{0k}=\frac{1}{\sqrt{d+1}}e^{i\alpha_k}$, $k=1,\dots,d-1$ has been used. Furthermore, the matrix $\tilde{K}$ is hermitian from (\ref{M0K_app}), so the phases $\alpha_j$ satisfy the symmetry relation $\alpha_{d-j}=-\alpha_j$, for every $j=1,\dots,d-1$. Note that this is the symmetry that allows to reduce the number of free parameters in fiducial states. Combining (\ref{M0K_app}) and (\ref{lambda_k_app}), an expression for the amplitudes of $\vec{\lambda}(M_0)$ in the computational basis arises:
\begin{equation}\label{ampli_lambda_app}
\{|(\vec{\lambda}(M_0))_0|^2,\dots,|(\vec{\lambda}(M_0))_{d-1}|^2\}=\sqrt{\frac{d}{d+1}}F^{\dag}\{\sqrt{d+1},e^{i\alpha_1},\dots,e^{i\alpha_{d-1}}\}^T.
\end{equation}
Therefore, there are $\lfloor(d-1)/2\rfloor+d-1$ free parameters in $\vec{\lambda}(M_0)$, where $\lfloor(d-1)/2\rfloor$ of them come from the absolute values (\ref{ampli_lambda_app}) and $d-1$ from the complex phases. The proof concludes by noting that $|\phi\rangle=\vec{\lambda}(M_0)/\sqrt{d}$ is a fiducial state, without loss of generality.   
\end{proof}
Prop. \ref{prop:lambdas} can be equivalently obtained by applying the discrete version of the Wiener-Khintchine Theorem \cite{W30} to the fiducial state $|\phi\rangle$. This simple observation has been surprisingly overlooked by the community, whereas the approach to MUF played a fundamental role to straightly arrive to Proposition \ref{prop:lambdas}. It is stressed that some numerical simulations made by Chris Fuchs in low dimensions suggest that the number of free parameters in fiducials might be $3(d-1)/2-1$ for $d$ odd and $3d/2-1$ for $d$ even, having an almost perfect match with Proposition \ref{prop:lambdas}, see page 1258 at Chris Fuchs's \emph{samizdat}  \cite{marcus_comment,F14}, the last paragraph in \cite{ADF14}, and Section IV in \cite{FHS17}.

\subsection{Minimum uncertainty states}
Fiducial states are minimum uncertainty states for a complete set of MUB in prime dimensions \cite{ADF14}. To show this property, let $\{|\varphi^j_k\rangle\}$ be a set of $d+1$ MUB in prime dimension $d$ and  $H_j=-\log_2[\sum_{k=0}^{d-1}(p^j_k)^2]$ be the quadratic R\'enyi entropy \cite{R61} associated to the $j$th MUB, where $p^j_k=|\langle\phi|\varphi^j_k\rangle|^2$. Thus, it holds that \cite{BW07}:
\begin{equation}\label{minH}
\sum_{m=0}^dH_m\geq (d+1)\log_2\left[\frac{d+1}{2}\right],
\end{equation}
for any pure quantum state $|\phi\rangle$. The saturation of inequality (\ref{minH}) defines a minimum uncertainty states if and only if $\sum_{k=0}^{d-1}(p^j_k)^2=\frac{2}{d+1}$, for every $j=0,\dots,d$. It has been shown that this last property is satisfied by any fiducial state $|\phi\rangle$ in every prime dimension $d$ \cite{ADF14}. However, there are further quantum states minimizing (\ref{minH}). An interesting question is the following: \textit{are there further sets of orthonormal bases such that fiducial states are minimum uncertainty states?} The following result, based on Prop. \ref{prop:lambdas}, allows to provide a positive answer to the question, in every dimension where a fiducial state exists.
\begin{prop}\label{product_prob}
Let $|\phi\rangle$ be a $d$-dimensional quantum state  satisfying (\ref{reduction}), $|\xi_0\rangle,\dots,|\xi_{d-1}\rangle$ the columns of any unitary matrix belonging to the Clifford group, and $p_k=|\langle\phi|\xi_k\rangle|^2$. Therefore,
\begin{equation}\label{min_uncertainty}
\sum_{j=0}^{d-1}p_jp_{j+k}=\frac{1+\delta_{k,0}}{d+1}.
\end{equation}
In particular, (\ref{min_uncertainty}) holds for fiducial states $|\phi\rangle$.
\end{prop}
\begin{proof}
For $\{|\xi_j\rangle\}$ equal to the computational basis, (\ref{min_uncertainty}) arises by taking the inner product between the left hand side of (\ref{reduction}) and the same expression after applying operator $X^k$. For the right hand side, it has to be used the fact that $X^kF^{\dag}=F^{\dag}Z^k$. Next, assume that $\{\xi_j\}$ are the columns of any unitary $U$ belonging to the Clifford group, and consider that $|\phi\rangle=U|\phi'\rangle$, where $|\phi'\rangle$ is another fiducial state. Thus, $p_k=|\langle\phi| k\rangle|^2=|\langle\phi| U^{\dag}U|k\rangle|^2=|\langle\phi'|\xi_k\rangle|^2$. The generality of the fiducial state $|\phi'\rangle$ allows to conclude that (\ref{min_uncertainty}) holds for any basis $|\xi_k\rangle=U|k\rangle$, $k=0\dots,d-1$, given by the columns of $U$.
\end{proof}
When considering the identity matrix as the chosen element of the Clifford group,  Proposition \ref{product_prob} reduces to Corollary 2.3 in \cite{K08}, also shown in eq.(10) in \cite{ADF14} and as a special case of Theorem 1 in \cite{F09}. Proposition \ref{product_prob} has the following immediate consequence.
\begin{corol}\label{corol:mus}
Any $d$-dimensional quantum state $|\phi\rangle$ that satisfies (\ref{reduction}) is a minimum uncertainty state with respect to every subset of $d+1$ orthonormal bases formed by the columns of $d+1$ elements of the Clifford group. In particular, the result holds for any fiducial state in every dimension $d$, and for $d+1$ MUB in every prime dimension $d$.
\end{corol}
\begin{proof}
From Prop. \ref{product_prob}, inequality (\ref{minH}) is saturated for any subset of $d+1$ bases, induced by the columns of Clifford unitaries. In particular, a maximal set of MUB arises from the columns of a subset of elements of the Clifford group in every prime dimension $d$. This is so, because there exists a subset of $d+1$ unitaries $U_0,\dots,U_{d}$ within the Clifford group that diagonalize operators $Z,X,XZ,XZ^2,\dots,XZ ^{d-1}$. Indeed, $Z$ belongs to the Clifford group, $F$ diagonalizes $X$, and the Clifford unitaries $U_{H_1},\dots,U_{H_{d-1}}$, where $H_k=\{\{k,-1\},\{-1,0\}\}$, produce $U_{H_k}XZ^kU^{\dag}_{H_k}=\tau^{-k}Z$, with $\tau=-e^{i\pi/d}$. In particular, the unique eigenvectors bases of operators $Z,X,XZ,XZ^2,\dots,XZ ^{d-1}$ form a maximal set of $d+1$ MUB in every prime dimension $d$  \cite{BBRV01}.
\end{proof}
\medskip

A relation between some classes of fiducial states and mutually unbiased bases can be shown. Assume that $|\phi\rangle$ is a fiducial state, eigenvector of the Zauner's operator
\cite{Z99} 
\begin{equation}\label{Zauner}
\mathcal{Z}=e^{i\pi(d-1)/12}FG,
\end{equation}
where $G$ is a diagonal unitary matrix having diagonal entries $G_{jj}=e^{i\pi(d+1)j^2/d}$, $j=0,\dots,d-1$. Thus, the following result arises.
\begin{prop}\label{triplet}
Any eigenvector of the Zauner operator (\ref{Zauner}) has identical probability distributions with respect to the triplet of MUB given by the columns of  $\{\mathbb{I},\mathcal{Z},\mathcal{Z}^2\}$, in every dimension $d$. In particular, the result holds for fiducial states that are eigenvectors of $\mathcal{Z}$.
\end{prop}
\begin{proof}
Zauner's operator $\mathcal{Z}$ is unitary, thus having
$|\langle k|\phi\rangle|=|\langle k|\mathcal{Z}^{\dag}\mathcal{Z}|\phi\rangle|=|\langle k|(\mathcal{Z}^2)^{\dag}\mathcal{Z}^2|\phi\rangle|$, 
for every $k=0,\dots,d-1$. Furthermore, if $|\phi\rangle$ is eigenvector of $\mathcal{Z}$ then $|\langle k|\phi\rangle|=|\langle k|\mathcal{Z}^{\dag}|\phi\rangle|=|\langle k|(\mathcal{Z}^2)^{\dag}|\phi\rangle|$, for every $k=0,\dots,d-1$. This implies that $|\phi\rangle$ has identical probability distributions in the three orthonormal bases $\{k\}$, $\{\mathcal{Z}|k\rangle\}$ and $\{\mathcal{Z}^2|k\rangle\}$. On the other hand, from (\ref{Zauner}) it is simple to check that $|\langle j|\mathcal{Z}|k\rangle|^2=|\langle j|\mathcal{Z}^{\dag}\mathcal{Z}^2|k\rangle|^2=1/d$ and $|\langle j|\mathcal{Z}^2|k\rangle|^2=1/d$, for every $j,k=0,\dots,d-1$. Therefore, the three bases are mutually unbiased.
\end{proof}
For $d=2$, any eigenvector of $\mathcal{Z}$ is a \emph{MUB-balanced} state \cite{WS07}, in particular when it is a fiducial state. The MUB-balanced property means that a $d$-dimensional quantum state has identical probability distributions with respect to a maximal set of $d+1$ MUB.\bigskip

After completing the work, it was noted that Falk Unger obtained (\ref{reduction}) in a slightly different way (unpublished notes \cite{Unger}). Namely,  if $|\phi\rangle=\sum_j c_j|j\rangle$ is a fiducial state then $e^{i\alpha_k}/\sqrt{d+1} = \langle\phi | Z^k |\phi\rangle = \sum_j |c_j|^2 \omega^{jk}=\sqrt{d}[F\sum_j|c_j|^2|j\rangle]_k$, for every  $k=1,\dots,d-1$.

\section{Conclusions}
This work introduced the ultimate notion of unbiasedness for sets of linearly independent and normalized vectors in $d$-dimensional Hilbert spaces: \emph{mutually unbiased frames} (MUF). These sets are characterized by a constant overlapping that can take any value between the absolute bounds 0 and 1. The overlap reduces to $1/d$ when at least one of the frames is tight, where tightness is equivalent to impose that the frame defines a rank-one positive operator valued measure (POVM). The introduced notion was illustrated by finding the most general pair of MUF composed by two vectors each for a qubit system. The concept of MUF allowed to think a Symmetric Informationally Complete (SIC)-POVM in dimension $d$ as a set of $d$ MUF composed by $d$ vectors each, thus having a common root with the mutually unbiased bases problem. Interestingly, Zauner's conjecture about the existence of fiducial states arose in a natural way when studying sets of $d$ MUF in dimension $d$ (see Prop. \ref{MZauner}). The approach also allowed to unlock two fundamental properties of SIC-POVM: real fiducial states do not exist in any even dimension (see Prop. \ref{real_fid}) and, any unknown $d$-dimensional fiducial state can be parameterized, a priori, with roughly $3d/2$ real variables (see Prop. \ref{prop:lambdas}). Additionally, some classes of quantum pure states, those satisfying (\ref{reduction}), are minimum uncertainty states with respect to $d+1$ orthonormal bases given by the columns of any subset of $d+1$ Clifford operations. This result holds for any fiducial state in every dimension $d$ where a Weyl-Heisenberg covariant SIC-POVM exists (see Prop. \ref{product_prob} and Corol. \ref{corol:mus}). Finally, any eigenvector of the Zauner's operator `looks the same' for three MUB, in every dimension $d$ (see Prop. \ref{triplet}). 

\vspace{1cm}
\textbf{Acknowledgements}
\medskip

The authors kindly acknowledge M. Appleby, A. Araya, I. Bengtsson, C. Fuchs, M. Grassl, F. Unger and K. \.Zyczkowski for their valuable comments. This work is supported by grant FONDECyT Iniciaci\'{o}n nro. 11180474, Chile. FCP and VGA belong to the first year of the PhD Program \emph{Doctorado en F\'isica, menci\'on F\'isica-Matem\'atica}, Universidad de Antofagasta, Antofagasta, Chile.
\end{document}